%% file: main.tex
\documentclass[letterpaper,twocolumn,10pt]{article}
\usepackage{usenix-2020-09}

\usepackage{blindtext}
\usepackage{multirow}
\usepackage{array}
\usepackage{graphicx}
\usepackage{booktabs}

\usepackage{enumitem}
\usepackage{subcaption}
\usepackage{bbm}
\usepackage{amsmath}
\usepackage{amsthm}

\newtheorem{theorembody}{Theorem}

\usepackage[ruled,linesnumbered]{algorithm2e}
\usepackage{bbm}
\usepackage{algorithmicx}
\usepackage{algpseudocode}

\usepackage{pifont}
\newcommand{\cmark}{\ding{51}}%
\newcommand{\xmark}{\ding{55}}%
\usepackage{authblk}

\begin{document}
\title{HeterMoE: Efficient Training of Mixture-of-Experts Models on Heterogeneous GPUs}

\date{}

\author{
\rm{Yongji Wu$^{\text{1}, *}$\enskip
    Xueshen Liu$^{\text{2}, *}$ \enskip
    Shuowei Jin$^{\text{2}}$ \enskip
    Ceyu Xu$^{\text{4}}$ \enskip\\}
\rm{Feng Qian$^{\text{3}}$ \enskip
    Z. Morley Mao$^{\text{2}}$ \enskip
    Matthew Lentz$^{\text{4}}$ \enskip
    Danyang Zhuo$^{\text{4}}$ \enskip
    Ion Stoica$^{\text{1}}$ \enskip}\\
{$^{\text{1}}$UC Berkeley\enskip$^{\text{2}}$Univerisity of Michigan\enskip$^{\text{3}}$University of Southern California\enskip$^{\text{4}}$Duke Univerisity}
}

\include{macros}

\maketitle

\input{abstract}
{\let\thefootnote\relax\footnote{{$^*$Yongji Wu and Xueshen Liu contributed equally.}}}

\input{intro}

\input{background}

\input{overview}
\input{design}

\input{implementation}

\input{evaluation}
\input{related}
\input{conclusion}

\bibliographystyle{plain}
\bibliography{reference}


\end{document}

%% file: macros.tex
\newcommand{\sys}{HeterMoE\xspace}
\newcommand{\parallelism}{zebra parallelism\xspace}
\newcommand{\Parallelism}{Zebra parallelism\xspace}
\newcommand{\asymfull}{asymmetric expert assignment\xspace}
\newcommand{\Asymfull}{Asymmetric expert assignment\xspace}
\newcommand{\asymabbr}{Asym-EA\xspace}
\newcommand{\asymstrategy}{gather and squeeze}

\newcommand{\attngpusym}{\text{Attn}}
\newcommand{\expgpusym}{\text{Exp}}

\def\Snospace~{\S{}}
\renewcommand*\sectionautorefname{\Snospace}
\def\sectionautorefname{\Snospace}
\def\subsectionautorefname{\Snospace}
\def\subsubsectionautorefname{\Snospace}
\newcommand{\autorefsuffix}[2]{\hyperref[#1]{\autoref*{#1}#2}}
\renewcommand{\algorithmautorefname}{Algorithm}

%% file: abstract.tex
\begin{abstract}
The Mixture-of-Experts (MoE) architecture has become increasingly popular as a method to scale up large language models (LLMs). To save costs, heterogeneity-aware training solutions have been proposed to utilize GPU clusters made up of both newer and older-generation GPUs. However, existing solutions are agnostic to the performance characteristics of different MoE model components (i.e., attention and expert) and do not fully utilize each GPU's compute capability.

In this paper, we introduce \sys, a system to efficiently train MoE models on heterogeneous GPUs. Our key insight is that newer GPUs significantly outperform older generations on attention due to architectural advancements, while older GPUs are still relatively efficient for experts. \sys disaggregates attention and expert computation, where older GPUs are only assigned with expert modules. Through the proposed \parallelism, \sys overlaps the computation on different GPUs, in addition to employing an asymmetric expert assignment strategy for fine-grained load balancing to minimize GPU idle time. Our evaluation shows that \sys achieves up to 2.3x speed-up compared to existing MoE training systems, and 1.4x compared to an optimally balanced heterogeneity-aware solution. \sys efficiently utilizes older GPUs by maintaining 95\% training throughput on average, even with half of the GPUs in a homogeneous A40 cluster replaced with V100. 

\end{abstract}

%% file: intro.tex
\section{Introduction}
Over the past few years, large language models (LLMs) have demonstrated impressive capabilities in domains like conversation agents~\cite{dam2024complete} and coding assistants~\cite{roziere2023code,team2024codegemma}. Recently, the sparsely-activated Mixture-of-Experts (MoE) architecture has gained popularity as the preferred way of scaling models to hundreds of billions of parameters~\cite{liu2024deepseek,yang2024qwen2}. MoE models are often trained with expert parallelism~\cite{lepikhin2020gshard} where the experts are distributed across GPUs, while token activations are exchanged between GPUs using all-to-all communication.

Meanwhile, GPU hardware is continuously evolving, with newer GPUs increasing performance but simultaneously cost.
The demand for newer GPUs is high, with the latest Grok~3~\cite{grok3-cost} and upcoming Llama~4~\cite{llama4-cost} requiring over 100K H100 GPUs (and costing over 4 billion)~\cite{100k-gpu-cluster}.
Moreover, the latest GPUs face significant supply constraints, often with a backlog of several months~\cite{nvidia-backlog,aws-genai-compute-2023}.

Therefore, we need to answer the following question: How can we effectively train MoE-based LLMs on clusters with multiple generations of GPUs?
Due to the high cost and limited supplies of latest GPUs, many organizations often retain nodes with older GPUs (e.g., V100) while adding new nodes with the latest GPUs (e.g., H100)~\cite{jayaram2023sia, benson2024cephalo} when upgrading infrastructure.
Given the extremely high resource requirement for LLM training, we need to utilize all available GPUs.

Leveraging heterogeneous GPUs is challenging due to their varying hardware properties (e.g., memory size, computation capability).
To split data, prior work has looked at scaling the batch size in data parallelism for each GPU~\cite{moreno2020training,kim2022scale}.
To split the model, prior work has explored unevenly distributing model layers in pipeline parallelism across GPUs~\cite{jia2022whale,um2024metis,yan2024flashflex}.
However, the fundamental limitation of these solutions is that they are agnostic to the existence of heterogeneity \emph{within} the model architecture itself.

Our insight is that different components of MoE models (i.e., attention and expert) exhibit distinct performance characteristics across GPU generations.
Older GPUs remain highly efficient for expert computation.
In contrast, attention performs significantly better on newer GPUs due to architecture-specific optimizations.
For instance, FlashAttention~v2~\cite{dao2023flashattention} exclusively supports Ampere and newer GPUs~\cite{flashattn-v100}, while FlashAttention~v3~\cite{shah2024flashattention} leverages Hopper-specific features like \texttt{wgmma} instructions and TMA~\cite{hopper-arch}. As experts are already placed across GPUs in expert parallelism, we are presented with an opportunity to assign each GPU only components it can efficiently compute, without introducing additional communication.

In this paper, we present \sys to efficiently train MoE models with heterogeneous GPUs.
\sys disaggregates the attention and expert blocks of a transformer layer, assigning them to two different generations of GPUs (newer and older, respectively).
The attention-expert disaggregation not only better harvests the compute power of older GPUs, but also alleviates the memory pressure on newer GPUs due to the dominant expert weights and their limited availability.

Still, there are two key challenges \sys must address. First, how can we overlap the computation of different GPUs to reduce the idle wait time of GPUs? Na\"ive attention-expert disaggregation leaves GPUs spending most of their time waiting for each other due to data dependency.
Second, to maximize the extent of overlapping, how can we balance the computation on each GPU at a fine granularity? Simply tuning the degree of parallelism leads to a narrow optimization space, limited by the number of valid configurations. For instance, in expert parallelism, the number of experts must be divisible by the number of GPUs to distribute them.

To address these challenges, we propose \parallelism, which divides an input batch to multiple micro-batches and overlaps the attention computation on the newer GPUs and expert computation on the older GPUs of different micro-batches. \Parallelism differs from pipeline parallelism, where the model is partitioned into multiple GPUs at the granularity of one or more layers. In pipeline parallelism, each sample is sequentially computed on each GPU with the corresponding layers, from the first GPU to the last one. In contrast, in \parallelism, the model is partitioned within a single layer, and each sample is computed in a zigzag pattern, passing back and forth between attention (newer) and expert (older) GPUs. To enable fine-grained load balancing of attention and expert GPUs, we propose an \asymfull (\asymabbr) mechanism, where we place a part of the experts back to attention GPUs when expert computation is slow. \asymabbr can be selectively activated for a subset of layers and moves back a different number of experts for different layers. We further develop a "\asymstrategy" strategy for \asymabbr to optimize each layer's assignment and minimize the GPU idle time, i.e., bubbles. 

We implemented \sys in PyTorch, evaluated across MoE models of different scales, using both an on-premise testbed and EC2 instances under different heterogeneity settings. Our results show that \sys outperforms existing MoE training systems by up to 2.3x and an optimally balanced heterogeneity-aware solution by 1.4x. In addition, \sys achieves 95\% training throughput on average compared to a homogeneous setting where all older GPUs are replaced by the newer ones.

We summarize our contributions as follows:
\begin{itemize}[leftmargin=*]
    \item  We observe that the performance disparity between newer and older generation GPUs differs significantly for attention and expert blocks, motivating our solution to disaggregate the two blocks for training MoE models on heterogeneous clusters, with no extra communication.
    \item We propose \parallelism to overlap not only the computation of attention GPUs with the computation of expert GPUs, but also the computation and all-to-all communication on each GPU.
    \item We design an \asymfull mechanism to selectively move different numbers of experts back to attention GPUs for different layers, enabling fine-grained load balancing for \parallelism to minimize bubbles.
\end{itemize}

%% file: background.tex
\section{Background and Motivation}

\subsection{Mixture-of-Experts Models}

\begin{figure}
\centering
\includegraphics[width=0.9\linewidth]{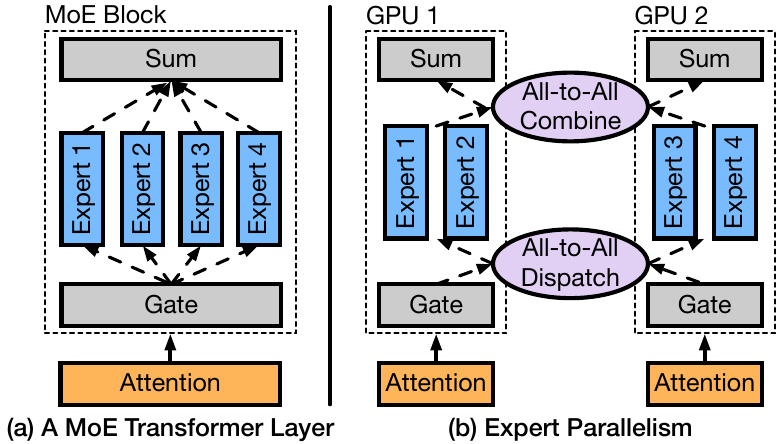}
\caption{MoE architecture and expert parallelism.}
\label{fig:moe_arch}
\end{figure}

The mixture-of-expert architecture, shown in \autoref{fig:moe_arch}, has been rapidly adopted in recent LLMs and has demonstrated as a cost-efficient solution for parameter scaling. In MoE models, the dense feed-forward network (FFN) in a transformer layer is replaced by multiple parallel FFNs, which are called experts. For each token, a trainable gate network takes the output embeddings of the attention block and computes a confidence score for each expert. Based on the scores, the token is only routed to the top-$k$ experts, whose outputs are weighted by their scores and summed up. The selective activation of experts enables MoE to scale model parameters with a sub-linear increase in computation costs. The attention blocks remain the same as dense LLMs. Unlike experts that mainly consist of GEMM operations, which are highly optimized on most GPUs, attention is less optimized. Only recently have implementations like FlashAttention~\cite{dao2022flashattention,dao2023flashattention,shah2024flashattention} been proposed, which are specifically tuned to recent GPU architectures and can significantly accelerate attention.

Expert weights have become increasingly dominant in the overall model size as sparsity increases. In Mixtral~8x7B~\cite{jiang2024mixtral}, 27.6\% of the parameters are active, while the number has decreased to 5.5\% in DeepSeek~v3~\cite{liu2024deepseek}. Expert parallelism is hence proposed for training MoE models, which distributes the experts of a MoE block to multiple GPUs. For each MoE block, a dispatch all-to-all collective is performed to send the embeddings of input tokens to GPUs with target experts, followed by a combine all-to-all that retrieves the outputs.

\subsection{Heterogeneous Training}
There is an increasing demand for training LLMs on heterogeneous GPU clusters due to frequent GPU release cycles, high upgrade costs and a lack of supply. Traditional training systems typically assume a homogeneous cluster and split workloads evenly across all GPUs. Under heterogeneous settings, such systems will let faster GPUs idle and wait for slower ones, significantly under-utilizing faster GPUs.  

To mitigate the bottlenecks caused by slower GPUs, existing heterogeneity-aware training systems like~\cite{jia2022whale,um2024metis,yan2024flashflex} unevenly split both data and models across different GPUs. They assign different batch sizes for different GPUs accordingly to balance the data load on each GPU, while to split the model, layers are partitioned according to each GPU's compute and memory capabilities.

However, these systems have several limitations for training MoE models. First, they do not differentiate different model components, failing to assign each component only to the GPUs that can efficiently execute it. Second, it is still infeasible to train large MoE models using only layer-based partitioning, as the weights of a single MoE block may exceed the GPU memory. Finally, since the model is partitioned at layer granularity, it can be challenging to find a partition that effectively balance the compute of different GPUs. Furthermore, balancing the compute often conflicts with memory capacities, hence limited by memory, faster GPUs may be assigned fewer layers required to fully utilize their compute capabilities, while slower GPUs may leave their memory underutilized due to their limited compute capabilities.

\begin{figure}
    \begin{subfigure}{0.46\linewidth}
        \centering
        \includegraphics[width=\linewidth]{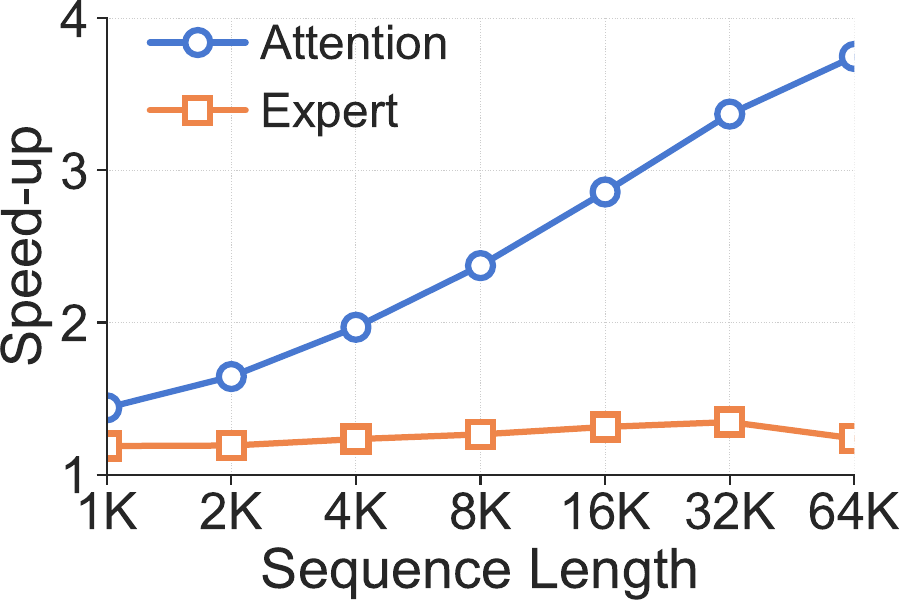}
        \caption{A40 over V100}
        \label{fig:background_speedup_a40_v100}
    \end{subfigure} \hfil
    \begin{subfigure}{0.46\linewidth}
    \centering
      \includegraphics[width=\linewidth]{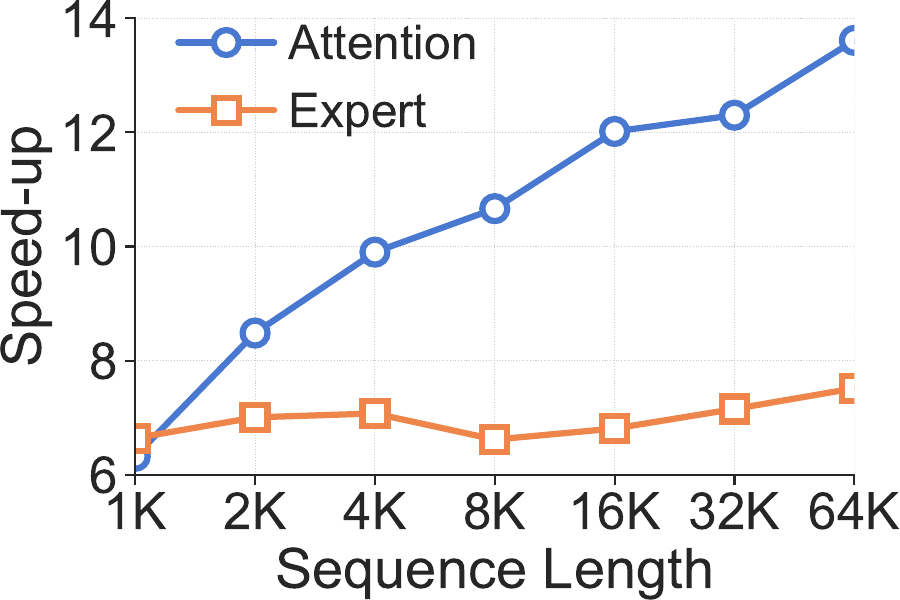}
        \caption{L40S over T4}
        \label{fig:background_speedup_l40_t4}
    \end{subfigure}
\caption{Speed-up of newer generation GPUs over older ones on attention and expert modules from Mixtral~8x7B~\cite{jiang2024mixtral}.}
\label{fig:background_speedup}
\end{figure}

\subsection{Opportunities}
We find that the relative efficiency of different GPUs differs significantly for different model components. In \autoref{fig:background_speedup}, we show the speed-up of two newer generation GPUs over two older ones. We measure each GPU's total forward and backward time for a single attention and expert module,using the model settings from Mixtral~8x7B~\cite{jiang2024mixtral}. We find that V100 GPUs are still quite capable for computing experts, achieving on average 80\% of A40's performance, which is equipped with the fully enabled flagship GA102 chip of Ampere generation. 
The relative performance of V100 compared to A40 on experts also remains stable across different sequence lengths, as each token is independently computed.

V100 suffers from much worse performance on attention modules as it does not support FlashAttention~\cite{flashattn-v100} and hence greatly bound by memory bandwidth, although we still use an optimized attention implementation~\cite{lefaudeux2022xformers}. The performance gap between A40 and V100 also rapidly widens as sequence length increases, contributed by the quadratic complexity of self-attention with respect to the sequence length. For 64K sequences, A40 outperforms V100 by 3.7x. From \autoref{fig:background_speedup_l40_t4}, since T4 is two generations older  with half the SMs of V100, L40S speeds-up MLP over T4 by 7.0x on average. Still, the speed-ups on attention is much more significant, with L40S outperforms T4 by 9.9x for 4K sequences and 13.6x for 64K sequences. We note that although T4 support FlashAttention~v1~\cite{dao2022flashattention}, due to limited SRAM sizes, it offers limited speeds-up and does not support larger attention head dimension~\cite{flashattn-t4} as used in \cite{jiang2024mixtral,yang2024qwen2,liu2024deepseek}.

Hence, to efficiently utilize older GPUs, we should avoid assigning them with attention operations. Existing heterogeneous training systems, however, are agnostic to such differences between attention and experts.

%% file: overview.tex
\section{Overview}

\begin{figure}
\centering
\includegraphics[width=0.95\linewidth]{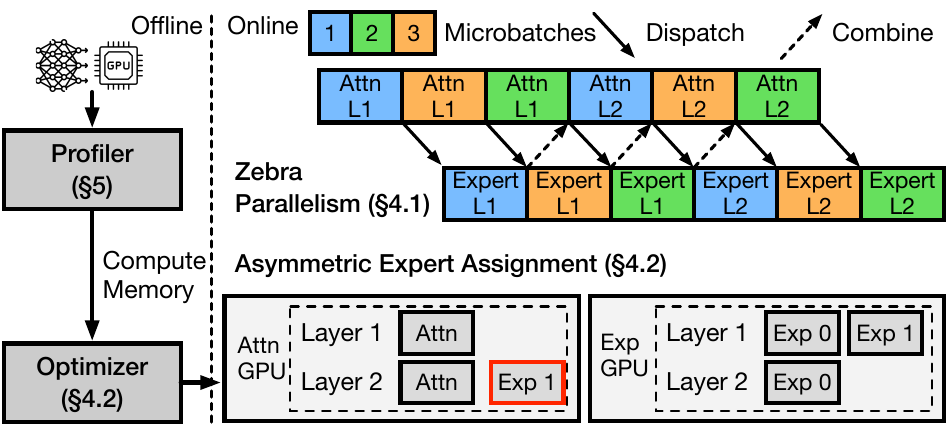}
\caption{Major components of \sys. \Parallelism overlaps the execution of attention and expert GPUs, while we introduce \asymabbr to minimize bubbles.}
\label{fig:overview}
\end{figure}

To efficiently train MoE models on heterogeneous GPU clusters, we propose \sys, a training system that disaggregates attention and expert computation within a transformer layer.
\sys assigns only expert blocks to older generation GPUs. Since MoE models are already trained using expert parallelism with all-to-all on homogeneous clusters, our disaggregation introduces no extra communication, as the total amount of data exchanged between GPUs remains unchanged, given the same global batch size. Such disaggregation not only improves the utilization efficiency of older GPUs, but also reduces memory pressure on newer GPUs by offloading expert weights to many older GPUs. It reduces the cost of MoE training by using fewer newer generation GPUs while having minimal performance degradation. 

We present the major components of \sys in \autoref{fig:overview}. We designate older generation GPUs that are only assigned with experts as expert GPUs, while newer GPUs are designated as attention GPUs.
Naively computing the attention and expert blocks on separate GPUs results in attention GPUs being idle while waiting for expert GPUs during expert computation (and vice versa).
To address this issue, we propose \parallelism (\autoref{sec:zp}), which divides each transformer layer into two stages, allowing attention and expert GPUs to simultaneously work on different microbatches.

To minimize the idle time, we need to closely balance the execution time of attention and expert GPUs for each microbatch.
We introduce \asymfull (\autoref{sec:asymea}), or \asymabbr, to enable fine-grained control of attention and expert computation across GPUs.
\asymabbr enables the migration of specific parts of the expert computation to (newer) attention GPUs; for instance, in \autoref{fig:overview}, Exp 1 for Layer 2 is moved from the Expert GPU to the Attention GPU to reduce bubbles in the compute stream.
Given this underlying mechanism for migration, we develop a "\asymstrategy" algorithm, which is implemented by the Optimizer (\autoref{sec:asymea}), to determine the best expert assignment.
The Optimizer depends on profiling information from the Profiler (\autoref{sec:impl}), such as the computation time and memory consumption for attention and expert blocks on the different GPUs.

%% file: design.tex
\section{Design}
In \autoref{sec:zp}, we first discuss how \parallelism works, under the base case where attention GPUs are not assigned with any experts. Next, in \autoref{sec:asymea}, we discuss how  \asymabbr can assign attention GPUs  with some experts to reduce bubbles.

\subsection{Zebra Parallelism}
\label{sec:zp}
\Parallelism (ZP) takes the place of traditional expert parallelism (EP) under heterogeneous settings. In an EP group, the experts of each layer are evenly split across $E$ GPUs, while attention blocks are replicated on each GPU. In contrast, in a ZP group, expert modules are distributed on $N$ expert GPUs, while all other components, including attention blocks and input/output embedding layers are replicated on $M$ attention GPUs. 
ZP differs from pipeline parallelism (PP) as PP assigns segments of consecutive layers to different GPUs, while ZP splits modules within each layer to different GPUs. ZP can be used in conjunction with PP, where we have ZP to split experts in each PP stage.

In ZP, computation on attention and expert GPUs are overlapped, as they process different microbatches at a time. Similar to EP. ZP still relies on all-to-all communication to exchange tokens.
The communication in ZP is bipartite. During the forward phase, each of the $M$ attention GPUs only sends tokens to the $N$ expert GPUs during dispatch, and only receives from expert GPUs during combine. During the backward phase, the communication is reversed. 

We consider network links between attention and expert GPUs have sufficient bandwidth so that dispatch and combine communication is faster than the attention and expert computation, similar to existing MoE training systems~\cite{li2023accelerating,hwang2023tutel}. We also assume a homogeneous network topology, otherwise, if the interconnect bandwidth within attention or expert GPUs is faster than the links between them, hierarchical all-to-all optimizations~\cite{hwang2023tutel,nie2022hetumoe} can be applied for \sys to take advantage of.
Still, all-to-all takes up to 30\%-50\% of the overall training time~\cite{hwang2023tutel,li2023accelerating,zhang2025comet}, hence
\sys also overlaps it with computation on each GPU.

\begin{figure}
\centering
\includegraphics[width=0.98\linewidth]{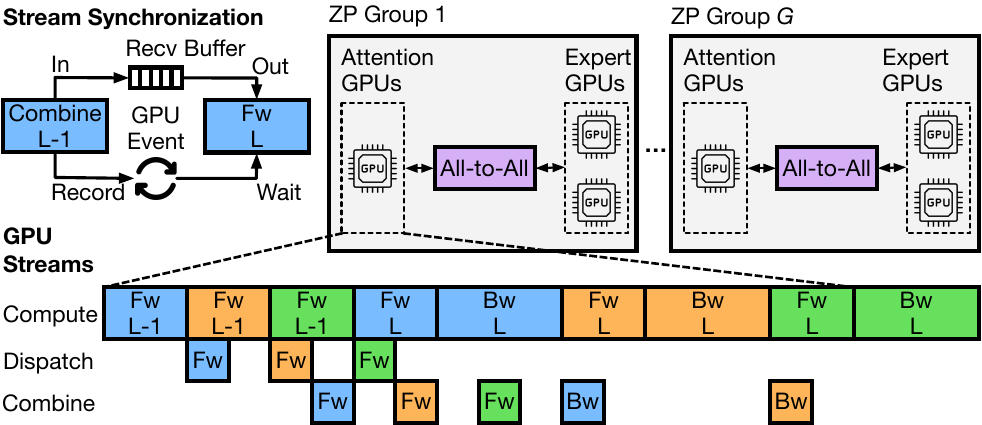}
\caption{Zebra parallelism replaces expert parallelism in heterogeneous settings. It overlaps compute on attention and expert GPUs, as well as compute and communication within each GPU.}
\label{fig:zp}
\end{figure}

To enable the overlapping, we use separate GPU streams for computation and communication. 
Since dispatch and combine all-to-all communicate in opposite directions and cause no bandwidth contention, they are scheduled on two independent communication streams and are overlapped. Hence, on each GPU, \sys maintains two streams for communication and one stream for computation. To maintain the data dependency between computation and communication tasks, \sys uses GPU events to synchronize different streams. For instance, to receive the input data, an all-to-all kernel is first launched (enqueued) on one of the communication stream, \sys then  records the corresponding event. The computation stream would wait for that event and block the execution of attention or expert computation until the all-to-all is finished and the data in the receive buffer is ready.

With zebra parallelism to overlap computation on different GPUs, as well as computation and communication on the same GPU, we still need to feed it with an execution schedule to minimize the per-iteration training time. We use $A^{F}_{i,j}, E^{F}_{i,j},D^{F}_{i,j},C^{F}_{i,j}$ and $A^{B}_{i,j}, E^{B}_{i,j},D^{B}_{i,j},C^{B}_{i,j}$ to denote the forward/backward attention computation ($A$), expert computation ($E$), dispatch ($D$) and combine ($C$) all-to-all tasks for the $j$-th microbatch on the $i$-th layer. A task can start its execution as long as its dependent task is finished, and all previous tasks on the same stream finishes. For instance, $A^{F}_{i,j}$ can start as long as:
\begin{equation*}
\left\{
\begin{aligned}
    t(A^{F}_{i,j}) \geq t(C^{F}_{i-1,j}) + T_C, && i \in{(1,L]}, j\in[1,R] \\
    \left|t(A^{F}_{i,j})-t(A^{F}_{i^\prime,j^\prime})\right|\geq T_{A}, && \forall(i^\prime,j^\prime)\neq(i,j)
\end{aligned}
\right.
,
\end{equation*}
where $t(\cdot)$ is the start time of a task, $L$ is the number of layers, $R$ is the number of microbatches. $T_A$ and $T_C$ are the duration of attention computation and combine all-to-all communication for a microbatch. The first constraint is for data dependency, the attention computation of a microbatch cannot start before the previous layer's expert outputs are received. The second constraint enforces sequential execution of a stream, i.e., the attention GPU can only compute a single microbatch at a time. Intuitively, an optimal schedule should enable each task to start as soon as possible.

\begin{figure}
\centering
\includegraphics[width=0.80\linewidth]{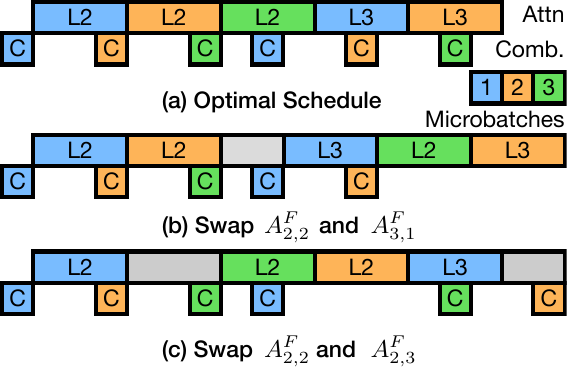}
\caption{Demonstration of the optimal execution schedule of \parallelism, compared to swapping any of the two tasks. We show the forward attention computation and combine all-to-all of the second and third layers.}
\label{fig:optimal_schedule}
\end{figure}

\begin{theorembody}
For a MoE model trained using \sys's \parallelism, where all experts are placed on expert GPUs while the execution of tasks follow data dependency and stream sequential execution constraints, the following execution schedule on each GPU minimizes the total time of a training iteration: \\
for computation on attention GPUs:
\begin{equation*}
(A^{F}_{1,1}\cdots A^{F}_{1,R})\cdots(A^{F}_{L,1}A^{B}_{L,1}\dots A^{F}_{L,R}A^{B}_{L,R})\cdots(A^{B}_{1,1}\cdots A^{B}_{1,R}),
\end{equation*}
for computation on expert GPUs:
\begin{equation*}
\begin{aligned}
&(E^{F}_{1,1}\cdots E^{F}_{1,R})\cdots(E^{F}_{L-1,1}\cdots E^{F}_{L-1,R})\\
&(E^{B}_{L-1,1}\cdots E^{B}_{L-1,R})\cdots(E^{B}_{1,1}E^{B}_{1,R}).
\end{aligned}
\end{equation*}
Following the above ordering of compute tasks, the dispatch and combine all-to-all tasks are scheduled to corresponding streams as soon as their dependent tasks are scheduled.
\label{theorem:zp_optimal_schedule}
\end{theorembody}

\begin{proof}
We prove that swapping any two tasks on the same stream will not reduce the total iteration time, regardless of how tasks on other streams are scheduled, as long as data dependencies are maintained. We show the proof for the computation stream on attention GPUs. The schedules for other streams can be similarly proved.

First, if we switch the order of $A^{F}_{i,j}$ and $A^{F}_{i+1, j^\prime}$, where $1\leq i  < L$ and $1\leq j^{\prime} <j\leq R$, the earliest possible start time $t^{\prime}(A^{F}_{i+1,j^\prime})$ after swapping will not be earlier than $t(A^{F}_{i,j})$ when the two tasks are not swapped, as $A^{F}_{i+1, j^\prime}$ needs to wait for the completion of $C^{F}_{i,j^\prime}$ and $A^{F}_{i,j-1}$. If the compute tasks on expert GPUs are not swapped, we have $t^\prime(C^{F}_{i,j^\prime}) >t(C^{F}_{i-1,j})$, otherwise we have either $t^\prime(C^{F}_{i,j^\prime}) \geq  t(C^{F}_{i-1,j})$ or $t^\prime(A^{F}_{i,j-1}) \geq  t(A^{F}_{i,j-1})$, depending on how the tasks are shuffled. Hence, the earliest start time of subsequent compute tasks will also be equal to or larger than their start time when $A^{F}_{i,j}$ is scheduled before $A^{F}_{i+1, j^\prime}$. \autoref{fig:optimal_schedule}(b) shows this case. Similarly, swapping $A^{B}_{i+1,j}$ with $A^{B}_{i, j^\prime}$ leads to sub-optimal schedules. 

Second, if we swap $A^{F}_{i,j}$ and $A^{F}_{i,j^\prime}$, where $1\leq i \leq L$ and $1 \leq j <j^{\prime} \leq R$, the earliest possible start time $t^{\prime}(A^{F}_{i,j^\prime})$ will be no earlier than the earliest start time $t(A^{F}_{i,j})$ when the two tasks are not swapped, since either $t^\prime(C^{F}_{i-1,j^\prime}) \geq t(C^{F}_{i-1,j})$ or $t^\prime(A^{F}_{i,j-1}) \geq t(A^{F}_{i,j-1} )$ (for $j=1$, it is $A^{F}_{i-1,R}$). \autoref{fig:optimal_schedule}(c) shows this case. Similarly, we can prove the total execution time will not decrease if we swap $A^{B}_{i,j}$ and $A^{B}_{i,j^\prime}$. 

Finally, if we swap $A^{B}_{L,j}$ with $A^{F}_{L,j^\prime}$ where $1\leq j < j^{\prime}\leq R$, the start time $t^{\prime}(A^{F}_{L,j^\prime})$ after swapping will not be earlier than $t(A^{B}_{L,j})$ when not swapped. All subsequent compute tasks will not start earlier. In particular, the execution of $A^{B}_{L,j}$ and $C^{B}_{L-1,j}$ will be deferred.
\end{proof}

With ZP to overlap the computation, we next explore how \sys further balances the computation of each micro-batch between the attention and expert GPUs in a fine-grained manner to minimize bubbles in a training iteration.

\subsection{Asymmetric Expert Assignment}
\label{sec:asymea}

\begin{figure*}
\centering
\includegraphics[width=0.99\linewidth]{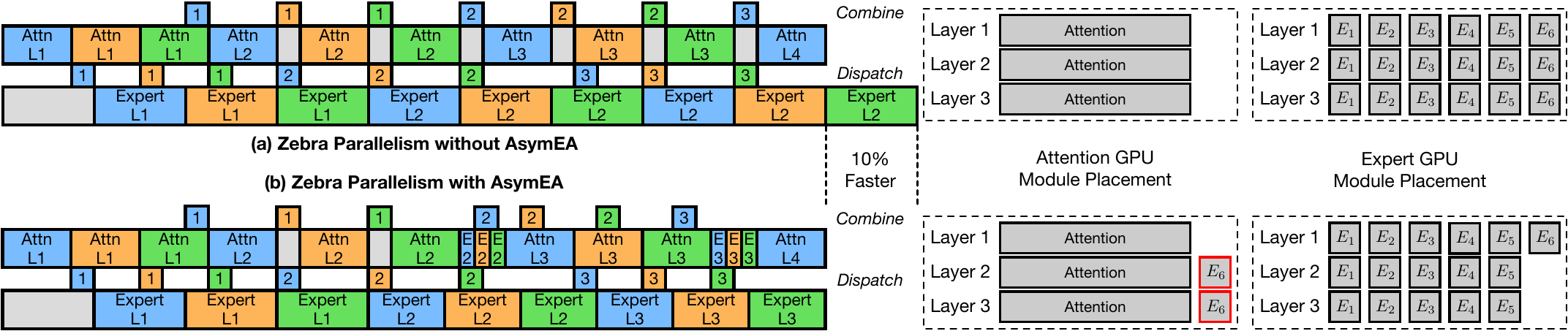}
\caption{\asymabbr enables fine-grained balance of computation load on attention and expert GPUs to minimize bubbles.}
\label{fig:asymea}
\end{figure*}

In \sys's \parallelism, the computation time on attention and expert GPUs for a single microbatch and a single layer is determined given the number of attention and expert GPUs in each ZP group, i.e., $M$ and $N$ respectively. Expert computation takes place on lower-end GPUs and typically dominates the overall computation, especially for shorter sequences. Therefore, attention GPUs may be left with a series of bubbles as a result of waiting for the next microbatch's expert computation (and combine all-to-all) to finish after completing their attention computation. In \autoref{fig:asymea}, we analyze a concrete scenario in which computation on expert GPUs is 33\% slower than on attention GPUs. Given three microbatches, we show the execution schedule for the forward pass of the first three layers. For default \parallelism (shown in \autoref{fig:asymea}a), despite the fact that expert GPUs are fully occupied by computation tasks, attention GPUs only spend 75\% time on effective computation. 

To increase the utilization of attention GPUs, we propose an \asymfull (\asymabbr) strategy where we offload some of the expert computation onto the (more powerful) attention GPUs. 
Given the ratio of experts to offload, \sys splits the experts across attention and expert GPUs in a ZP group,  During dispatch and combine, tokens are sent to and received from both attention and expert GPUs. On each attention GPU, \sys separates attention and expert computation of each microbatch, as it must receives tokens from other attention GPUs before computing experts.

In \autoref{fig:asymea}(b), we put one of the experts of the 2nd and 3rd layer to attention GPUs, while expert GPUs only handle five.  For each layer, we schedule the expert computation on attention GPUs after attention computation of all microbatches.
In this case, although the attention GPUs still have bubbles before the 2nd layer's attention computation, all subsequent computation becomes bubble-free. The 60\% reduction in bubbles directly translates to 10\% speed-up of the total forward time of the first 3 layers.

However, for a single layer, even if we only offload a single expert in this example, the computation on attention GPUs for each microbatch becomes longer than that on expert GPUs. If we simply offload the same number of experts for every layer, bubbles may shift to expert GPUs. For example, if we also offload an expert in the first layer in \autoref{fig:asymea}(b), the total execution time does not decrease, as expert GPUs would then suffer from 6\% bubbles. Therefore, we need to selectively offload specific layers.

To balance the total workloads on attention and expert GPUs, hence reducing the bubbles, we propose an algorithm based on a "\asymstrategy" strategy to determine the set of layers to offload experts, and the number of experts to offload.
Our high-level insight is to "gather" enough bubbles on attention GPUs across several consecutive layers until we can offload at least one expert to "squeeze" the bubbles. 
We present this algorithm in \autoref{algo:asymea}.

\input{algorithm}

%% file: algorithm.tex
\begin{algorithm}[t]
\SetAlgoLined
\SetAlgoNoEnd
\DontPrintSemicolon
\SetKwInOut{Input}{Input}
\SetKwInOut{Output}{Output}
\newcommand\myCommentStyle[1]{\small\textcolor{gray}{#1}}
\SetCommentSty{myCommentStyle}
\SetKwComment{Comment}{// }{}

\Input{$n$: Number of experts of each layer; $L$: Number of layers; $M,N$: Number of attention GPU and expert GPUs in a ZP group; $T^{\attngpusym}_{A}, T^{\attngpusym}_{E},T^{\expgpusym}_{E}$: Compute time of attention modules and expert modules on attention and expert GPUs.}
\Output{
$\mathcal{O}:\{o_1,o_2,\dots,o_L\}$: Number of experts to offload from each expert GPU at each layer. 
}
    
$n_1 \gets \max(1, \frac{N}{M})$\;
\Comment{\#experts each attention GPU acquires for a single chunk}
\label{eq:line_chunk1}
$n_2 \gets n_1 \cdot \frac{M}{N}$\; \Comment{\#experts each expert GPU offloads for a single chunk}
\label{eq:line_chunk2}
$T_{\text{gather}}\gets T^{\expgpusym}_{E}-T^{\attngpusym}_{A}$\;
\Comment{bubble at each layer without \asymabbr}
\label{eq:t_gather}
$T_{\text{squeeze}} \gets \frac{T^{\expgpusym}_{E}N}{n}n_1 + \frac{T^{\attngpusym}_{E}N}{n}n_2 $\;
\Comment{bubble eliminated by offloading a chunk of $n_2$ experts}
\label{eq:t_squeeze}
$t_{\text{bubble}} \gets 0$ \Comment{total accumulated bubble}
\For{$l \gets 1$ to $L$}{ \label{eq:start_for_loop}
    $t_{\text{bubble}} \leftarrow t_{\text{bubble}} + T_{\text{gather}} $\; \label{eq:acc_bubble}
    \If{$t_{\text{bubble}} \geq T_{\text{squeeze}}$}{
         \Comment{gather enough bubbles until we can squeeze}
         $o_l \gets \lfloor t_{\text{bubble}} / T_{\text{squeeze}} \rfloor$ \Comment{\#chunks to offload}
         \label{eq:chunks_offload}
         $t_{\text{bubble}} \gets t_{\text{bubble}} - o_l \cdot T_{\text{squeeze}}$\;
         $o_l \gets o_l \cdot n_2$
         \label{eq:experts_offload}
    }
}
\Return{$\mathcal{O}$}
\caption{\asymabbr offload optimization.}
\label{algo:asymea}
\end{algorithm}

First, given the ZP group setup of $M$ and $N$, we profile (\autoref{sec:impl}) the forward compute time of a single layer's attention module for a microbatch on attention GPUs, $T^{\text{Attn}}_{A}$, and a single layer's expert module on expert GPUs, $T^{\text{Exp}}_{E}$. 
We note that $T^{\text{Exp}}_{E}$ depends on the amount of tokens each expert GPU processes instead of the experts it is assigned.
We also profile the compute time of a single expert (FFN) with the same batch size on attention GPUs as $T^{\attngpusym}_{E}$. The profiled time is fed into \autoref{algo:asymea} as inputs.

Since each expert GPU and attention GPU must offload or acquire the same number of experts to ensure they have consistent workloads, and the total number of experts offloaded must equal that acquired by all attention GPUs, we define in line~\ref{eq:line_chunk1}-\ref{eq:line_chunk2} the minimum number of experts $n_1$, that each attention GPU must acquire, and the minimum number of experts $n_2$ each expert GPU must offload. $n_1$ and $n_2$ forms the minimal chunk (unit) for offloading. We can only offload one or multiple chunks at each layer. To make sure $n_1$ and $n_2$ are integers, we assume either $M$ is a multiple of $N$, or $N$ is a multiple of $M$. The limitation only applies when \asymabbr optimization is used, and is similar to the restriction of traditional EP that the number of GPUs in a EP group must divides the number of experts. We also assume we have enough microbatches to overlap the communication on expert GPUs and fully saturate their computation. 

In line~\ref{eq:t_gather}, we compute $T_{\text{gather}}$, the bubble formed at each layer for a microbatch when \asymabbr is not used. $T_{\text{squeeze}}$ in line~\ref{eq:t_squeeze} is the bubble we can shrink by offloading a single chunk of $n_2$ experts from each expert GPU. $N\cdot T^{\expgpusym}_{E}/n$ is the compute time reduced on each expert GPU by offloading a single expert, as it receives fewer tokens. $N\cdot T^{\attngpusym}_{E}/n$ is the extra compute time added to each attention GPU, for each expert it acquires. We gather the bubbles in line~\ref{eq:acc_bubble} across multiple layers, until they are large enough to offload at least a single chunk. In line~\ref{eq:chunks_offload}, we compute how many chunks we can offload, and multiply by $n_2$ in line~\ref{eq:experts_offload}, we get the number of experts to offload from each expert GPU. 


We follow the forward pass to squeeze the bubbles and use profiled forward compute time as inputs. We note that backward time for each module scaled proportionally to the forward time, the assignment optimized from forward pass reduces both forward and backward time. 

\noindent
\textbf{Addressing memory limitations.}
In practice, however, the number of experts to offload are limited by the memory capacities of both attention and expert GPUs. On the one hand, if the $N$ expert GPUs cannot hold all experts across all layers, \sys must offloads some experts, and the total of experts to offload, $\sum{\mathcal{O}}$ will be lower bounded by $n_{\min}$. On the other hand, $\sum{\mathcal{O}}$ will be upper bounded by $n_{\max}$, as attention GPUs cannot hold too many experts. To take account of such limitations, we modify line~\ref{eq:acc_bubble} as follows:

\begin{equation*} \label{eq:alpha_beta_t}
t_{\text{bubble}} \leftarrow t_{\text{bubble}} + \alpha\cdot\beta \cdot T_{\text{gather}} 
\end{equation*}
where
\begin{equation*}
\begin{aligned}
\alpha &= \min\left(\frac{\lfloor n_{\max} /n_2\rfloor \cdot T_{\text{squeeze}}}{L \cdot T_{\text{gather}}}, 1\right)\\
\beta &= \max\left(\frac{\lceil n_{\min} /n_2\rceil  \cdot T_{\text{squeeze}}}{L \cdot T_{\text{gather}}}, 1\right).
\end{aligned}
\end{equation*}

Without the coefficients $\alpha$ and $\beta$, the amount of bubble we can gather over all $L$ layers is $L\cdot T_{\text{gather}}$, and we would offload  $\lfloor\frac{L\cdot T_{\text{gather}}}{T_{\text{squeeze}}}\rfloor$ chunks. If the number of chunks exceeds $\lfloor n_{\max} /n_2\rfloor$, i.e., the maximum allowed by the attention GPU's memory, we add a coefficient $\alpha < 1$ to $T_{\text{gather}}$. $\alpha$ enforces the upper bound, while line~\ref{eq:start_for_loop}-\ref{eq:experts_offload} spreads the offload chunks across layers. Similarly, $\beta$ enforces the lower bound of $\lceil n_{\min} /n_2\rceil$ chunks. Note that we have either $\alpha=1$ or $\beta=1$, since at most one of them is activated when $\lfloor\frac{L\cdot T_{\text{gather}}}{T_{\text{squeeze}}}\rfloor$ goes beyond the lower or upper bounds.

%% file: implementation.tex
\section{Implementation}
\label{sec:impl}
We implement \sys in 3K lines of Python based on PyTorch~v2.2~\cite{imambi2021pytorch}, with components from 
DeepSpeed~v0.14~\cite{rasley2020deepspeed}. 

\noindent
\textbf{\Parallelism engine.} Given a user-defined MoE model, using the MoE block provided by \sys and the ZP group setup, our \parallelism engine splits the model in each ZP group as in \autoref{sec:zp} and manages the training. During initialization, the ZP engine setups the three streams for compute and communication and allocates the receive buffers for each microbatch. It also establishes required collective communication groups. For all-to-all, \sys creates separate groups for dispatch and combine to enable concurrent communication on separate streams, with each group containing all GPUs in the EP group. \sys feeds unequal split sizes to PyTorch's NCCL all-to-all wrapper to distribute different number of tokens to different GPUs, based on the expert assignment. We implement ZP and data parallelism in our prototype, but \sys can be extended to support ZP in conjunction with other parallelization strategies, e.g., heterogeneity-aware pipeline parallelism~\cite{yan2024flashflex,um2024metis}. 

We note that the gate network generates a confidence score for each of the top-$k$ experts, which is used to compute weighted sum of the expert outputs. Such computation paradigm forms a "residual" connection. The backward pass from the MoE block's outputs is hence extended to two branches, one of them  propagates to confidence scores and the gate network's weights. This branch would propagate all the way back to the first layer, as it does not depend on the gradients from expert GPUs. To prevent repetitive backward on the same weights, \sys stops the backward propagation at attention outputs for each layer. The backward of previous layers would not start until the gradients of attention outputs (expert inputs) are received from expert GPUs and are accumulated with the gradients from the second branch.

\noindent
\textbf{Profiler.} We implement a profiler for \sys to obtain the compute time that is required by the \asymabbr optimizer as inputs. The profiler extracts a single expert FFN from a transformer layer. Given the global batch size, sequence length, the number of microbatchs and the ZP group setup, the profiler determines the number of tokens $B$ each expert GPU computes for a single microbatch. It then generates random tensors with a batch size of $B$ and feds it to the FFN to profile the expert compute time on both an attention ($T^{\attngpusym}_{E}$) and an expert GPU ($T^{\expgpusym}_{E})$. The remaining part of the transformer layer, including attention blocks and MoE gate, is extracted and profiled on an attention GPU to get $T^{\attngpusym}_{A}$, where the input size is set to the size of a microbatch.

Our profiler also measures the memory usage of the GPU. On an expert GPU, it constructs a single expert FFN, and executes forward and backward passes with dummy inputs. It then measures the memory usage, which contains activations, weights, gradients, and optimizer states. Based on the remaining available memory, the profiler estimates how many experts in total can expert GPUs hold. Subtracting it from the total number of experts, we get $n_{\min}$, the minimal number of experts \sys needs to offload. Similarly, on an attention GPU, the profiler constructs the model by excluding all expert modules. The forward and backward pass is executed by replacing all-to-all with dummy operations. The profiler then estimates $n_{\max}$, the maximum number of experts each attention GPU can hold. We only need to run the profiler once for each setup, and it only requires a single attention GPU and a single expert GPU.

%% file: evaluation.tex
\section{Evaluation}
\subsection{Setups}
\label{sec:eval_setups}

\begin{table}[t]
\caption{Cluster settings used in the evaluation. O refers to our on-premise cluster and C refers to AWS.}
\label{tab:cluster-setup}
\newcolumntype{P}[1]{>{\centering\arraybackslash}p{#1}}
\begin{tabular}{|P{18mm}|P{8mm}|P{8mm}|P{8mm}|P{8mm}|P{8mm}|}
\hline
Setup & \texttt{O1} & \texttt{O2} & \texttt{O3} & \texttt{C1} & \texttt{C2} \\
\hline
\multirow{2}{*}{\#GPUs} & 6xA & 4xA & 6xA & 2xL & 2xL\\
\cline{2-6}
 & 6xV & 8xV & 3xV & 6xT & 8xT\\ 
\hline
\multirow{2}{*}{GPU Model} & \multicolumn{3}{c|}{A: A40~(48GB)} & \multicolumn{2}{c|}{L: L40S~(48GB)}\\
\cline{2-6}
& \multicolumn{3}{c|}{V: V100~(16GB)} & \multicolumn{2}{c|}{T: T4~(16GB)}\\
\hline
Network & \multicolumn{3}{c|}{100~Gbps} & \multicolumn{2}{c|}{200~Gbps$^{{\ast}}$}\\
\hline
\end{tabular}
\end{table}

\noindent
\textbf{Cluster setups.}
We evaluate \sys on both an on-premise testbed and on AWS. On each testbed, we configure cluster setups with different numbers of GPUs of each type. \autoref{tab:cluster-setup} lists the setups we use. The GPUs in our on-premise testbed are connected with 100~Gbps Mellanox ConnectX-6 RoCE NICs. On AWS, we use g4dn.4xlarge and g6e.4xlarge instances, where they have 20~Gbps TCP networks. However, we find that under 20~Gbps links, communication takes up to 70\% of the total training time. As a token's communication time is much longer than its compute time for both attention and experts, it is impossible to overlap compute and communication, resulting in all compared methods bottlenecked by communication. 
Therefore, to match L40S's compute speed, we simulate a network connection of 200~Gbps on AWS by reducing the amount of data transferred in all-to-all.

\begin{table}[t]
\caption{Configurations of models used in the evaluation.}
\label{tab:model_comparison}
    \centering
    \resizebox{0.97\columnwidth}{!}{
    \begin{tabular}{lccccc}
        \toprule
        Model & \#Layers & Hidden Dim & \#Experts & \#Params\\ 
        \midrule
        Mixtral-W1 & 4 & 2048 & 12 & 2.2B \\
        \midrule
        Mixtral-W2  & 4 & 2048 & 24 & 4.3B \\
        \midrule
        Mixtral-D1  & 8  & 1024 & 24  & 2.1B \\
        \midrule
        Mixtral-D2 & 6 & 1024 & 18 & 1.2B \\
        \midrule
        Mixtral-D3  & 8 & 1024 & 40 & 3.5B \\
        \bottomrule
    \end{tabular}
    }
\end{table}

\begin{figure*}
\centering
\includegraphics[width=0.97\linewidth]{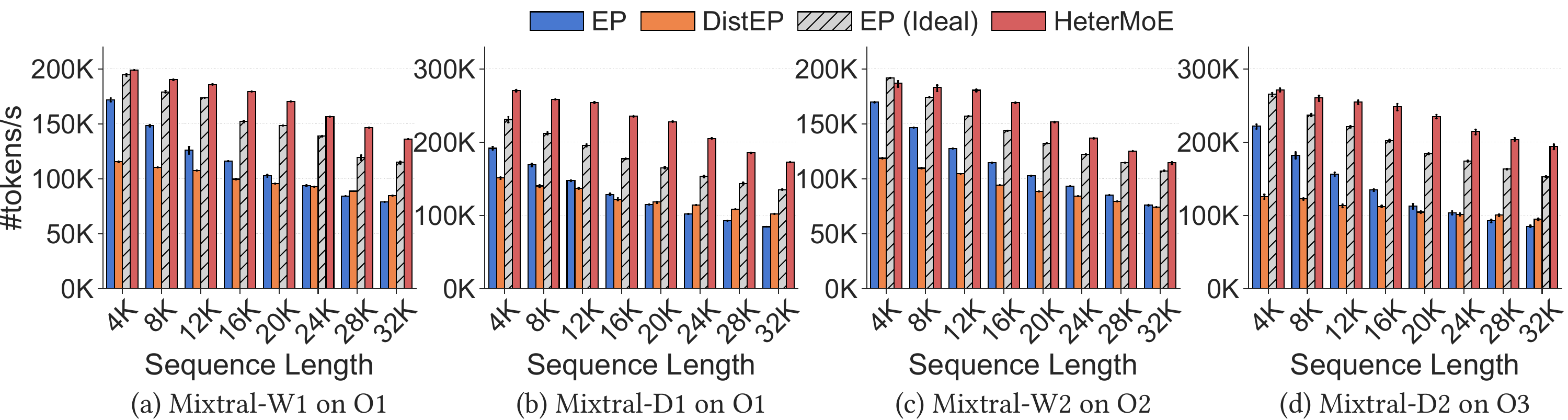}
\caption{\textbf{[Overall Performance]:} Overall training throughput under different sequence lengths on the on-premise cluster of A40 and V100 GPUs, with different models and GPU setups. Error bars represent 95\% confidence intervals.}
\label{fig:overall_on-premise}
\end{figure*}
\begin{figure}
\centering
\includegraphics[width=0.98\linewidth]{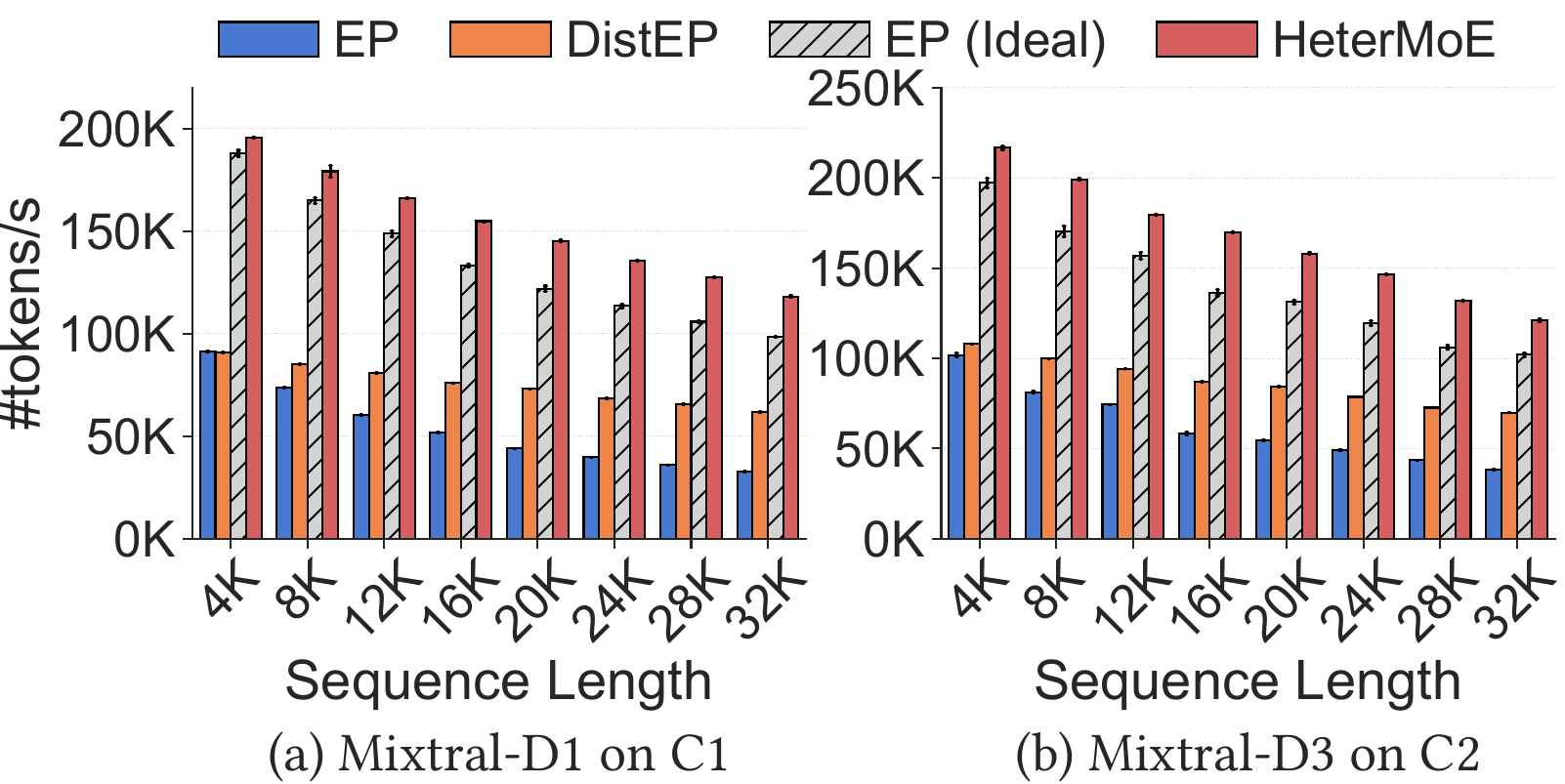}
\caption{\textbf{[Overall Performance]:} Overall training throughput on AWS with L40S and T4 GPUs.}
\label{fig:overall_aws}
\end{figure}

\noindent
\textbf{Models.}
We choose five MoE models of varying sizes based on the recent Mixtral architecture~\cite{jiang2024mixtral}, listed in \autoref{tab:model_comparison}. We explore both deeper models with more layers and wider models with larger hidden sizes. We also vary the number of experts to fit clusters of different sizes, as the number of GPUs (to distribute experts) must divide it. \cite{li2023accelerating}, we use top-2 gating for all models as in \cite{li2023accelerating,hwang2023tutel}. Following the standard practice~\cite{singh2023hybrid,chen2016training,narayanan2021efficient}, we apply activation checkpointing and FP16 mixed precision training. For each model and cluster setup, we set the global batch size to the maximum allowed by the GPU memory to run all baselines. 

\noindent
\textbf{Baselines.} As our \parallelism replaces traditional expert parallelism in heterogeneous environments, we mainly compared \sys with EP-based training under a single ZP/EP group. We implement state-of-the-art system-side MoE training optimizations from Tutel~\cite{hwang2023tutel} and Lina~\cite{li2023accelerating} on top of the widely adopted DeepSpeed~MoE~\cite{rajbhandari2022deepspeed}, which include optimized kernels and communication overlapping. 

Based on the optimized DeepSpeed~MoE, we compare two solutions: run it directly on the heterogeneous cluster using all GPUs, denoted as \textbf{EP}; naively disaggregate attention and expert modules without ZP to overlap the computation on different GPUs, denoted as \textbf{DistEP}.
Since there is no heterogeneity-aware EP solution that achieves optimal load balancing, we also compare with an ideal case by running DeepSpeed~MoE separately for each GPU model using all GPUs of that model, and then summing up their throughput. We denote it as \textbf{EP~(Ideal)}. For instance, for \texttt{O1} setup in \autoref{tab:cluster-setup}, we independently run DeepSpeed~MoE on 6xA100 and on 6xV100, then we simply add the throughput together. EP~(Ideal) assumes prefect load balancing and does not take account of communication overheads across different GPU models. It represents the theoretical maximum throughput a MoE training system without attention-expert disaggregation can achieve. 
We tune both the number of microbatches $R$ for \sys and the all-to-all partitioning degree for all compared methods to maximize performance. 

In addition, even though ZP/EP is orthogonal to data, tensor, pipeline parallelism and can be used in combination with them, we also compare \sys with pure heterogeneity-aware pipeline parallelism in \autoref{sec:eval_pipe} to show PP's limitations. 

\subsection{Overall Performance}
We report the overall training throughput in \autoref{fig:overall_on-premise} for our on-premise testbed, with sequence lengths from 4K to 32K. We note that long-context reasoning capabilities have been increasingly sought after, and recent LLMs are trained with context lengths of 32K or even up to 128K~\cite{liu2024deepseek,dubey2024llama,yang2024qwen2,glm2024chatglm}.

Since attention has quadratic computation complexity to the sequence lengths, the training throughput of all compared methods drops as sequence length increases. We observe that at shorter sequence lengths, EP still achieves decent performance, since A40 offers limited speed-ups on both attention and experts as shown in \autoref{fig:background_speedup_a40_v100}. For instance, for 4K sequences, EP maintains 82\% of \sys's performance on average, and it even reaches 91\% for Mixtral-W2 on \texttt{O2}. Since EP is not heterogeneity aware and does not balance the compute loads on A40 and V100, it suffers poor performance on longer sequences. For 16K sequences, \sys outperforms EP by 1.67x on average, and it further increases to 1.89x at 32K. In particular, for deeper but narrower Mixtral-D1 and D2, where attention is more dominant and attention-expert disaggregation is more essential, \sys's speed-ups over EP reach 2.05x and 2.29x for 32K sequences.

We also find that naively assigning experts to older GPUs without compute overlapping yields poor performance. DistEP achieves only 55\% of \sys' throughput on average across all settings. However, DistEP works relatively better for longer sequences when A40 has significant speed-ups on attention. On average, at 4K, DistEP's throughput is only 56\% that of \sys. It performs even worse than EP by 32\%. At 32K, DistEP increases to 59\% of \sys's performance. It slightly outperforms EP by 8\% on average, and up to 21\%.

Comparing \sys with the ideal case, EP~(Ideal), \sys still achieves an average speed-up of 1.18x. The speed-up is up to 1.38x for Mixtral-D1 on \texttt{O1} for 20K sequences, where \sys can match A40's attention compute time with V100's expert compute time.
For shorter sequences, attention is faster and \sys balances the computation on A40 and V100 with \asymabbr, which we provide a detailed breakdown in \autoref{sec:eval_asymea_breakdown}. Still, for O2 on Mixtral-W2, \sys suffers 3\% performance drop compared with EP (Ideal). This is due to the the limited speed-up of attention on A40 at 4K lengths, while expert computation is more intensive for the wider model variants. The limited depth and number of experts per GPU for Mixtral-W2 also reduce the search space for \asymabbr. In shallower models, the penalty of ZP from the 1st microbatch's first forward and first backward is more pronounced, where V100 has to wait for A40.

We show the results on AWS in \autoref{fig:overall_aws}. \sys still consistently outperforms all compared methods. Compared with EP (Ideal), \sys achieves an average speed-up of 1.17x and up to 1.25x. Since the performance gap between L40S and T4 is much larger than that between A40 and V100, EP and DistEP experience severe performance degradation as L40S remains idle most of the time. On average, \sys outperforms EP by 2.89x and DistEP by 1.96x.

Since we use a simulated network connection of 200~Gbps on AWS, unless otherwise specified, we perform subsequent evaluations on our on-premise testbed.  

In general, we find that the benefit of \sys over baselines is more significant for deeper models and models with more experts, matching recent model designs~\cite{liu2024deepseek,yang2024qwen2} that use many smaller experts instead of a few large ones. Although we only evaluate sequence lengths up to 32K due to limited GPU resources and memory constraints, we observe that the speed-up of \sys over EP also grows with increasing sequence lengths, which greatly benefits the training of long-context models.

\subsection{Comparing with Pipeline Parallelism}
\label{sec:eval_pipe}
\begin{figure}
\centering
\includegraphics[width=0.98\linewidth]{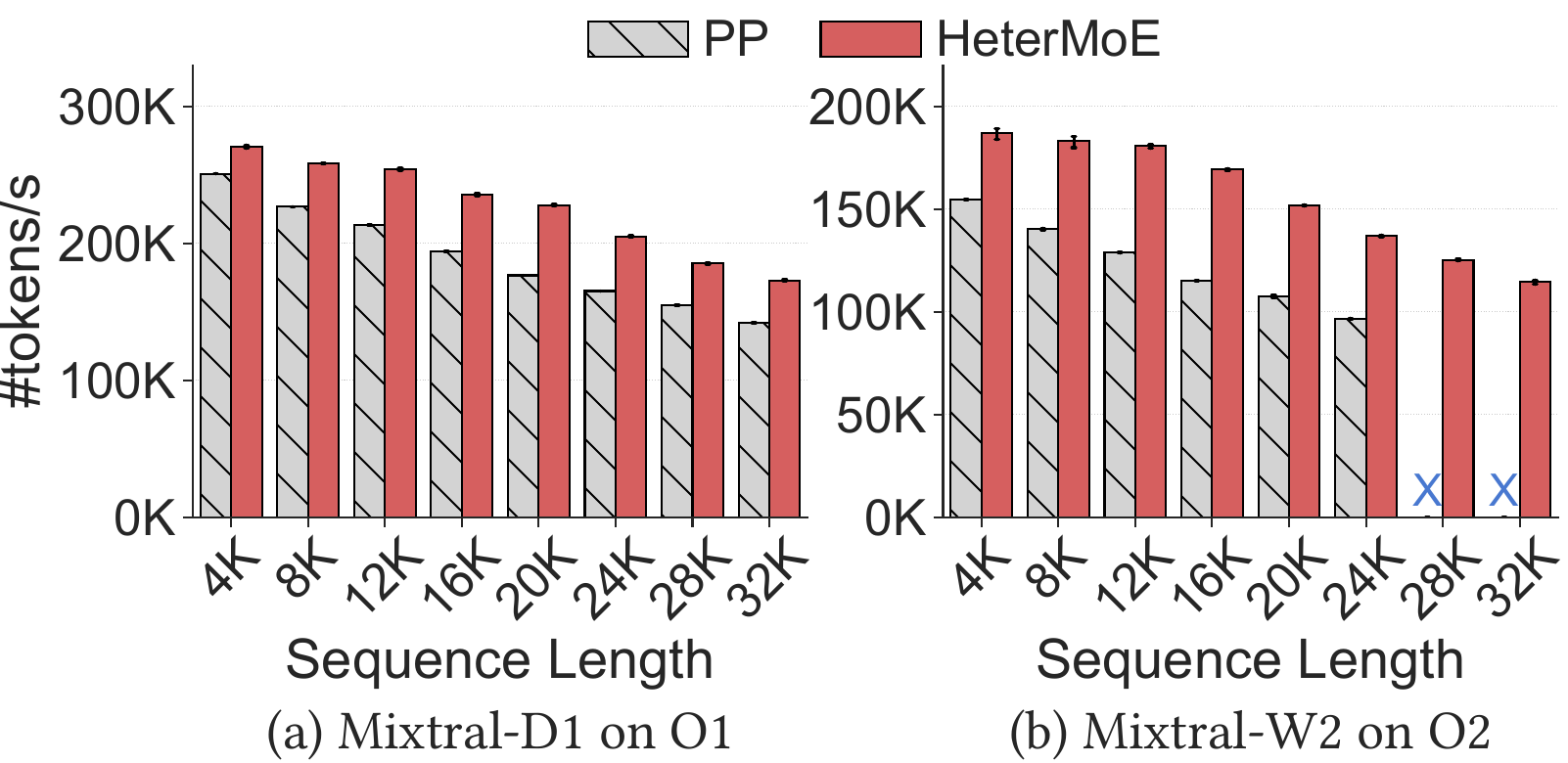}
\caption{\textbf{[Pipeline Parallelism]:} Performance comparison of \sys's \parallelism with heterogeneity-aware pipeline parallelism on the on-premise testbed. }
\label{fig:pipeline_on-premise}
\end{figure}

In this section, we study how our \parallelism compares to existing heterogeneity-aware training techniques~\cite{um2024metis,yan2024flashflex,jia2022whale}. They are mainly based on pipeline parallelism, where different pipelines stages are assigned to different GPU models. Each stage is also assigned with a different number of layers to balance the compute time, where faster GPUs are assigned with more layers. 

In \autoref{fig:pipeline_on-premise}, we present the training throughput of \sys and heterogeneity-aware pipeline parallelism (PP) on our on-premise testbed. 
For setup \texttt{O1}, we spin up 6 pipeline instances, with each containing a single A40 and a single V100. 
Similarly, for setup \texttt{O2}, we use 4 pipeline instances, each with 1x~A40 and 2x~V100. 
Training data is distributed across pipeline instances. Each GPU in a pipeline instance corresponds to a pipeline stage. We tune the layers assigned to each stage to balance the load and maximize the throughput, under the memory limitation posed by each GPU. We perform stage balance tuning independently for each model and each sequence length. 

We find that \sys consistently outperforms PP across all sequence lengths, achieving an average speed-up of 1.28x and up to 1.47x. \sys even outperforms PP by 7\%-21\% for 4K sequences, where attention-expert disaggregation provides marginal gains. PP's poor performance is due to several limitations it faces. 
First, PP does not distinguish between attention and expert modules.
Second, the granularity of PP's load balancing is restricted to a single layer, while \sys enables fine-grained load balancing for ZP with \asymabbr.
Finally, limited by the GPU memory, PP may fail to achieve an optimal layer assignment that balances the compute of each stage. 
As PP splits model by layers, a GPU may not fit even a single MoE block. In \autoref{fig:pipeline_on-premise}(b), PP cannot fit even a single layer and the activations of a single 28K or 32K sequence into V100's memory.
Hence, using PP alone to split a large MoE model is often not enough, it must be used in conjunction with expert parallelism or \sys's \parallelism to split experts within a layer. 

\subsection{Ablation Study}
\subsubsection{Impacts of GPU ratios in a ZP group}
\label{sec:eval_ablation_gpu_ratio}
\begin{figure}
\centering
\includegraphics[width=0.97\linewidth]{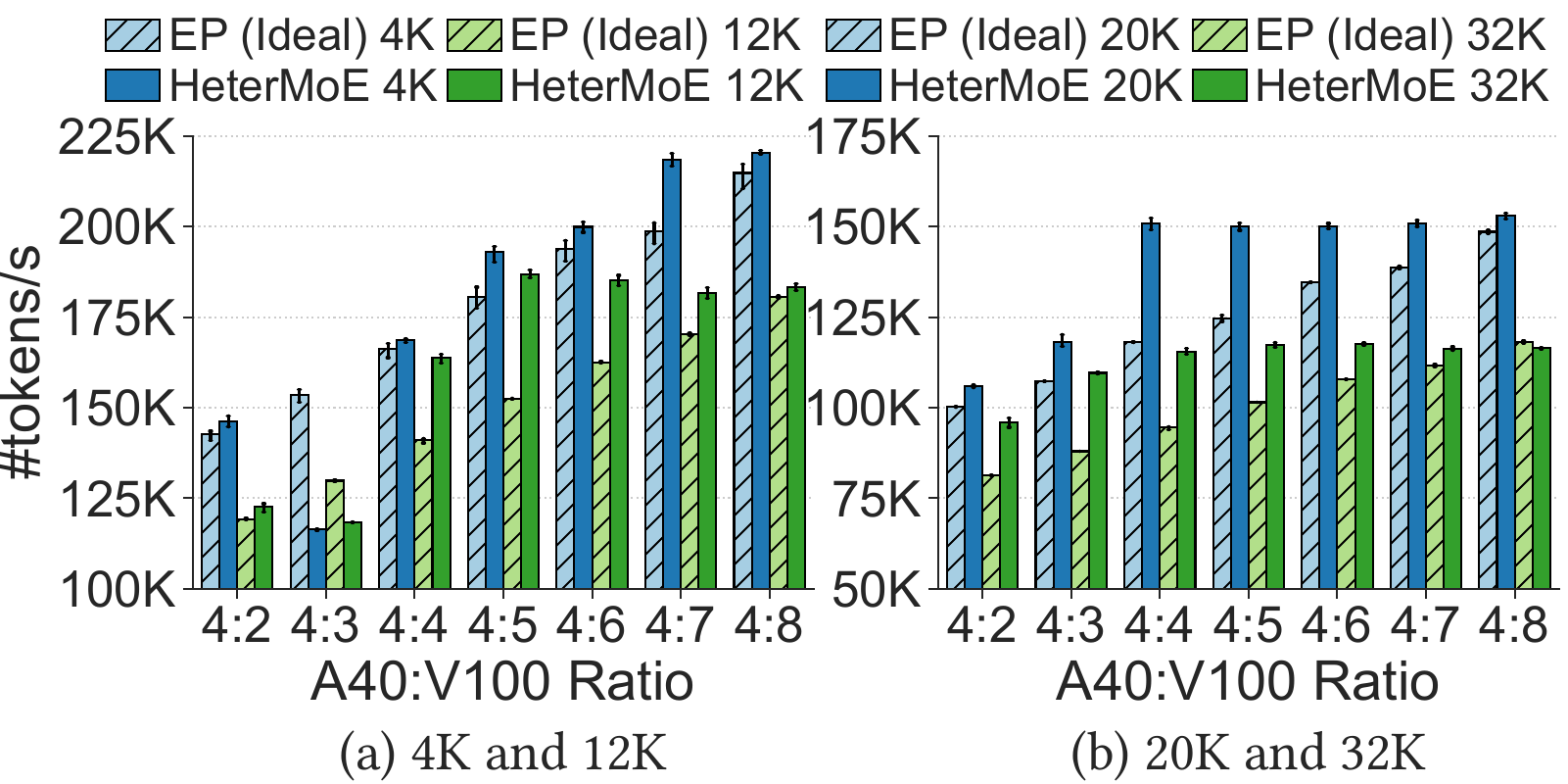}
\caption{\textbf{[Ablation Study]:} Impacts of GPU ratios in the setup of a ZP group to \sys's effectiveness.}
\label{fig:ablation_curve}
\end{figure}

Next, we study how the GPU ratio in a ZP group impacts \sys's performance. We change the ratio of A40 and V100 GPUs used in a single ZP group to study the best GPU ratio for \sys under different sequence lengths. The amount of compute and communication on each GPU depends solely on the GPU ratio, not their absolute number. Hence, to control the ratio, we fix the number of A40 GPUs to 4 while changing the number of V100 GPUs. We compare \sys's throughput with EP (Ideal) at different sequence lengths under each setup. We use Mixtral-D1 model architecture, but we scale the total number of experts linearly with the number of V100 to ensure that we can evenly distribute experts to GPUs for \sys as well as EP (Ideal). 
We present the results in \autoref{fig:ablation_curve}.

We find that the optimal A40 to V100 ratio varies for different sequence lengths. For example, \sys's effectiveness peaks at 4:5 for 12K sequences with a speed-up of 1.22x over EP (Ideal), while the peak for 4K sequences is at 4:7 with a speed-up of 1.10x. For longer sequences of 20K and 32K, \sys's throughput under 4:4 even reaches that of EP (Ideal) under 4:8 with 2x the number of V100s used, with a difference within 2\%.  At a fixed sequence length, \sys's throughput does not necessarily increase by simply increasing the relative proportion of expert GPUs (V100) in a single ZP group, as attention may instead dominate the compute time. We also note that \asymabbr is only effective at 4:2, 4:4 and 4:8, due to the divisibility requirement in \autoref{sec:asymea}.
This leads to a significant 24\% performance drop compared to EP (Ideal) under 4:3 for 4K sequences, while \sys is 3\% faster under 4:2.
For a target sequence length, a ZP group should be configured using the optimal GPU ratio with enough expert GPUs to hold all experts, while setting up multiple such ZP groups to utilize all available GPUs. We have also implemented a simulator to estimate the training throughput under different ZP group setups, where in addition to compute time, we also profile the communication time of NCCL send/recv under different message sizes. \unskip\footnote{PyTorch's all-to-all is implemented using NCCL send/recv operations.}

\subsubsection{Comparing with fully homogeneous setups}
\begin{figure}
\centering
\includegraphics[width=0.98\linewidth]{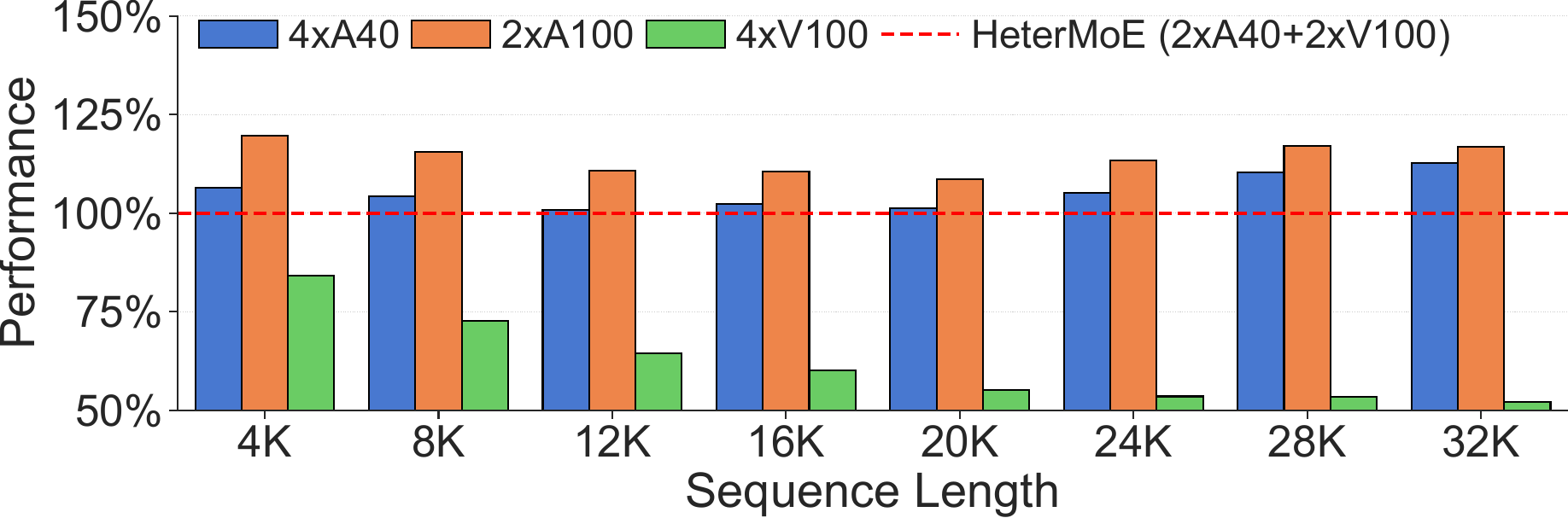}
\caption{\textbf{[Ablation Study]:} \sys's performance comparison with EP on fully homogeneous setups.}
\label{fig:ablation_full_homo}
\end{figure}

We also study how \sys compares to running EP, i.e., DeepSpeed~MoE on fully homogeneous cluster setups. We compare the performance of \sys using 2xA40 and 2xV100, to that of EP on 4xA40, 4xV100 and 2xA100 (80~GB). The 2xA100 are connected with PCIe Gen4. We use Mixtral-D1 model, but we set the total number of experts to 8 to match the number of GPUs. In \autoref{fig:ablation_full_homo}, we report the relative training throughput of EP compared to \sys, under different sequence lengths.

We note that an A100 delivers 2.1x FP16 tensor TFLOPS than an A40 while having 2.8x memory bandwidth~\cite{a40-datasheet,a100-datasheet}. Still, 2xA100 achieves only up to 1.20x speed-up over \sys, and is only 1.14x on average. Since V100 has performance similar to A40 for computing experts according to \autoref{fig:background_speedup_a40_v100}, \sys can efficiently harvest V100's compute. With half of the A40 replaced by V100, \sys still achieves 95\% the performance of 4xA40 on average. The performance gap is as little as 1\%-2\% for 12K-20K sequences, where \sys realizes decent load balancing. V100's inefficiency on attention leads to the poor performance of 4xV100. With 2xA40 and 2xV100, \sys achieves 1.66x speed-up over 4xV100 on average. Although for 4K sequences, 4xV100 still reaches 84\% the performance of \sys, it drops to 52\% for 32K sequences.

\subsubsection{Effects of \asymfull}
\label{sec:eval_asymea_breakdown}
\begin{figure}
    \begin{subfigure}{0.48\linewidth}
        \centering
        \includegraphics[width=\linewidth]{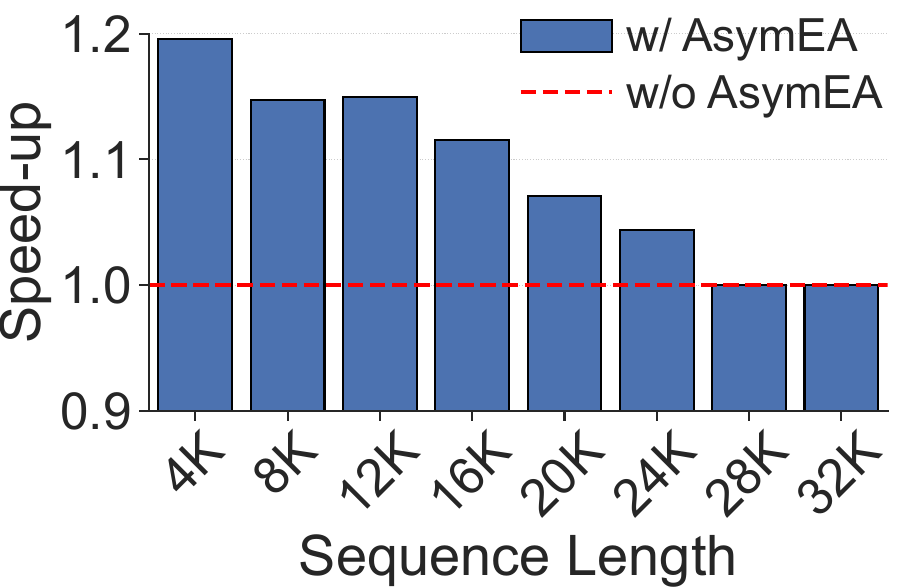}
        \caption{Mixtral-W1 on O1}
        \label{fig:ablation_asymea_w1_o1}
    \end{subfigure} \hfil
    \begin{subfigure}{0.48\linewidth}
        \centering
          \includegraphics[width=\linewidth]{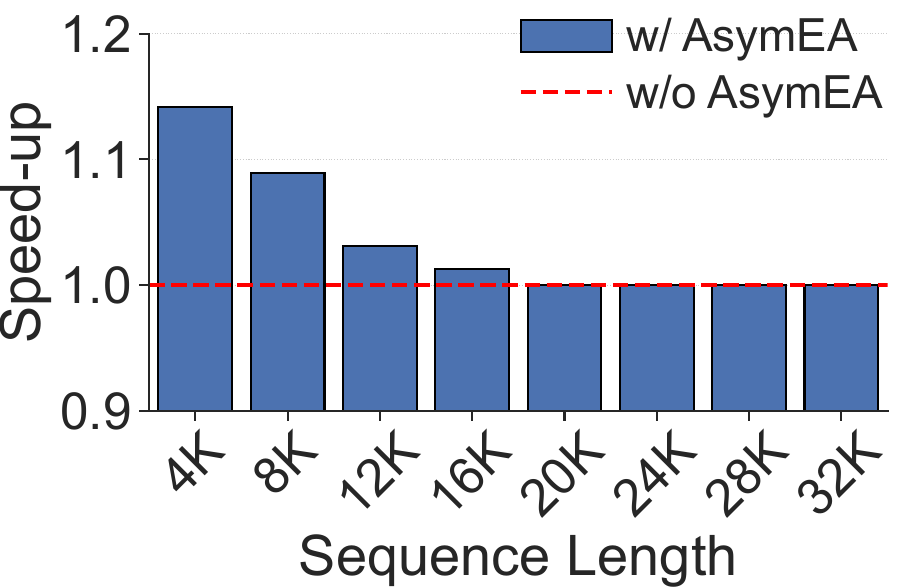}
        \caption{Mixtral-D1 on O1}
        \label{fig:ablation_asymea_d1_o1}
    \end{subfigure}
\caption{\textbf{[Ablation Study]:} Speed-up provided by \sys's \asymabbr in terms of training throughput, compared to \sys without \asymabbr.}
\label{fig:ablation_asymea}
\end{figure}

\begin{table}[t]
\caption{\textbf{[Ablation Study]:} Impacts of \sys's ZP (w/o \asymabbr) and \asymabbr to GPU utilization (percentage of time spent on effective compute) for Mixtral-D1 on \texttt{O1} setup. We also include the utilization improvement against DistEP in the parentheses.}
\label{tab:asymea_bubble_breakdown}
\centering 
\begin{tabular}{cccc}
    \toprule
    Seq. Len & \asymabbr & A40 Util. & V100 Util.  \\
    \midrule
    \multirow{2}{*}[-2pt]{8K} & \xmark & 55\% (1.69x) & 89\% (1.73x) \\    
    \cmidrule{2-4}  & \cmark & 83\% (2.51x) & 78\% (1.52x) \\
    \midrule
    \multirow{2}{*}[-2pt]{16K} & \xmark & 77\%  (1.90x) & 89\% (1.97x) \\  
    \cmidrule{2-4}  & \cmark & 86\% (2.12x) & 84\% (1.87x) \\  
    \bottomrule
\end{tabular}
\end{table}

We study the effectiveness of \asymabbr in \autoref{fig:ablation_asymea} on \texttt{O1} setup, where we compare the speed-up brought by \asymabbr for two settings under different sequence lengths. \asymabbr is most effective for shorter sequences where attention's compute time $T^{\attngpusym}_{A}$ on A40 is significantly faster than the compute time $T^{\expgpusym}_{E}$ of experts on V100. For 4K sequences, \asymabbr provides 1.20x speed-up on Mixtral-W1 and 1.14x on Mixtral-D1. As the gap between $T^{\attngpusym}_{A}$ and $T^{\expgpusym}_{E}$ closes with increasing sequence lengths, the additional contribution of \asymabbr gradually reduces. For instance, the speed-up of \asymabbr decreases to 1.04x on Mixtral-W1 at 24K. For sequences longer than 20K on Mixtral-D1 and sequences longer than 28K on Mixtral-W1, \asymabbr is no longer required, as we have $T^{\attngpusym}_{A} \geq T^{\expgpusym}_{E}$. Instead, we should increase the number of A40 in a ZP group to decrease $T^{\attngpusym}_{A}$.

We also provide the breakdown of how \sys improves GPU utilization, with and without \asymabbr. We show GPU utilization and the improvement compared to DistEP in \autoref{tab:asymea_bubble_breakdown}. We find that with ZP's ability to overlap computation on attention and expert GPUs, both A40 and V100's utilization are greatly improved, as DistEP without such overlapping spends 30\%-70\% of training time on idle waiting. Since \asymabbr moves some expert computation to A40, A40's utilization is significantly increased, at the expense of a small decrease in V100's utilization. 

%% file: related.tex
\section{Related Work}
\noindent
\textbf{MoE training systems.}
Extensive literature has been proposed to specifically optimize MoE training with expert parallelism. MegaBlocks~\cite{gale2023megablocks} proposes a grouped GEMM kernel to accelerate expert computation. A series of work~\cite{he2022fastermoe,nie2023flexmoe,wu2024lazarus,zhai2023smartmoe} optimize expert placement to handle the dynamic loads on experts. All-to-all communication is also optimized, with both collective implementation optimizations~\cite{hwang2023tutel,zhang2025comet,deepep2025,shi2024schemoe,nie2022hetumoe}and compute-communication overlapping~\cite{li2023accelerating,zhang2025comet,shi2024schemoe}. In particular, DeepSeek~\cite{liu2024deepseek} introduces DualPipe to overlap computation and communication within a pair of forward and backward chunks, while optimized kernels with tuned SM allocations are used for cross-node communication. \sys, on the other hand, disaggregates attention and experts by taking advantage of their differences in performance characteristics. Most of these optimizations are orthogonal and can be incorporated into \sys.

\noindent
\textbf{Heterogeneity-aware training systems.}
Training LLMs on heterogeneous clusters requires partitioning both data and model. 
Whale~\cite{jia2022whale} balances the data load within a pipeline stage by considering heterogeneous compute and memory, but it only addresses uneven memory demands across pipeline stages by considering device assignment. Metis~\cite{um2024metis} and FlashFlex~\cite{yan2024flashflex} further balance layers across stages. SDPipe \cite{miao2023sdpipe} targets dynamic heterogeneity where GPUs suffer from uncontrollable performance variations.
HAP~\cite{zhang2024hap} unevenly distributes workloads in both data and tensor parallelism. However, tensor parallelism requires frequent synchronization and HAP cannot overlap communication.
Cephalo~\cite{benson2024cephalo} targets FSDP and only balances compute by adjusting batch sizes. Recently, HEXA-MoE~\cite{luo2024hexa} is proposed for MoE training on heterogeneous GPUs. However, similar to Cephalo, it only balances compute by unevenly splitting data. 
These approaches are complementary to \sys and can be combined with zebra parallelism.

%% file: conclusion.tex
\section{Conclusion}
This paper presents \sys, a system for efficient Mixture-of-Experts (MoE) models training on heterogeneous GPUs. 
\sys disaggregates attention and expert modules to fully utilize each GPU's capability. \sys introduces \parallelism (ZP), along with \asymfull (\asymabbr), to enable computation overlapping and fine-grained load balancing. Our evaluations show that \sys consistently outperforms existing techniques, achieving up to 2.3x speedup over existing MoE training systems, while maintaining an average 95\% throughput with half GPUs of newer generation in a homogeneous cluster replaced by older ones. We will open source \sys.

\newpage